\documentclass[11pt]{article}%{ctexart}

\usepackage{graphicx}
\usepackage[margin=0.5in]{geometry}
\usepackage{bm}
\usepackage{amsthm}
\usepackage{amscd}
\usepackage{amsbsy}
\usepackage{appendix}
\usepackage{amsmath}
\usepackage{amsfonts}
\usepackage{amssymb}
\usepackage{calligra}
\usepackage{xcolor}
\usepackage{float}
\usepackage[T1]{fontenc}
\usepackage{multirow}
\usepackage{mathrsfs}
\usepackage{mathtools}
\usepackage{epstopdf}
\usepackage{psfrag}
\usepackage{subfigure}
\usepackage{booktabs}
\usepackage{listings}
\usepackage{longtable}
\usepackage{latexsym}
\usepackage{arydshln}
\usepackage{enumerate}
\usepackage{epsfig}
\usepackage{cite}
\usepackage{verbatim}
\usepackage{overpic}
\usepackage{abstract}
\usepackage{extarrows}
\usepackage[pagewise]{lineno}%\linenumbers

\usepackage[noend]{algpseudocode}
\usepackage{algorithmicx,algorithm}

\usepackage[figuresleft]{rotating}

\allowdisplaybreaks[4]

\date{}
\title{Jacobi Stability Analysis for Systems of ODEs Using \\ Symbolic Computation}

\author{Bo Huang$^{\text{a}}$, Dongming Wang$^{\text{b,c}}$ and Jing Yang$^{\text{d}}$\\
	\it\footnotesize $^{\text{a}}$LMIB -- School of Mathematical Sciences,
Beihang University, Beijing 100191, China \\
	\it\footnotesize bohuang0407@buaa.edu.cn\\
	\it\footnotesize $^{\text{b}}$LMIB -- Institute of Artificial Intelligence, Beihang University, Beijing 100191, China\\
    \it\footnotesize $^{\text{c}}$Centre National de la Recherche Scientifique, 75794 Paris Cedex 16, France\\
	\it\footnotesize Dongming.Wang@buaa.edu.cn\\
\it\footnotesize $^{\text{d}}$SMS -- HCIC -- School of Mathematics and Physics, Guangxi Minzu University, Nanning 530006, China \\
	\it\footnotesize yangjing0930@gmail.com}

\newtheorem {theorem*}{Theorem}
\newtheorem{theorem} {Theorem}

\newtheorem{definition}{Definition}

\newtheorem{proposition}{Proposition}
\newtheorem{lemma}{Lemma}

\newtheorem{remark}{Remark}

\newtheorem{open problem} {Open problem}

\numberwithin{equation}{section}

\usepackage[colorlinks,
            CJKbookmarks=true,
            linkcolor=blue,
            anchorcolor=red,
             urlcolor=blue,
            citecolor=blue
            ]{hyperref}

\begin{document}
\maketitle
\noindent {\bf Abstract.}  The classical theory of Kosambi--Cartan--Chern (KCC) developed in differential geometry provides a powerful method for analyzing the behaviors of dynamical systems. In the KCC theory, the properties of a dynamical system are described in terms of five geometrical invariants, of which the second corresponds to the so-called Jacobi stability of the system. Different from that of the Lyapunov stability that has been studied extensively in the literature, the analysis of the Jacobi stability has been investigated more recently using geometrical concepts and tools. It turns out that the existing work on the Jacobi stability analysis remains theoretical and the problem of algorithmic and symbolic treatment of Jacobi stability analysis has yet to be addressed. In this paper, we initiate our study on the problem for a class of ODE systems of arbitrary dimension and propose two algorithmic schemes using symbolic computation to check whether a nonlinear dynamical system may exhibit Jacobi stability. The first scheme, based on the construction of the complex root structure of a characteristic polynomial and on the method of quantifier elimination, is capable of detecting the existence of the Jacobi stability of the given dynamical system. The second algorithmic scheme exploits the method of semi-algebraic system solving and allows one to determine conditions on the parameters for a given dynamical system to have a prescribed number of Jacobi stable fixed points. Several examples are presented to demonstrate the effectiveness of the proposed algorithmic schemes.

\smallskip

\noindent {\bf Math Subject Classification (2020).} 34C07; 68W30.

\smallskip


\section{Introduction}

Differential equations are widely used in science and engineering for modeling real-world phenomena of various kinds. The problem of analyzing the qualitative behaviors of systems of differential equations is essential and has to be studied theoretically and/or practically whenever the dynamical properties of the solution of the systems in question need be well understood. In this paper, we are concerned with dynamical systems of first-order ordinary differential equations
\begin{equation}\label{eq1-0}
\begin{split}
\dot{\boldsymbol{x}}=\boldsymbol{f}(\boldsymbol{x},\boldsymbol{\mu}),\quad \boldsymbol{f}:\mathbb{R}^n\times\mathbb{R}^p\rightarrow\mathbb{R}^n,
\end{split}
\end{equation}
where $\boldsymbol{x}=(x_1,\ldots,x_n)$ are variables, $\boldsymbol{\mu}=(\mu_1,\ldots,\mu_p)$ are real parameters, and $\boldsymbol{f}=(f_1,\ldots,f_n)$ are rational functions in $\mathbb{R}(\boldsymbol{x})$. Our study is focused on describing the local stability of equilibria and the global stability of late-time deviations of solution trajectories. The global stability of the solutions of the dynamical systems can be described by using the well-known stability theory of Lyapunov, where the fundamental quantities are the so-called Lyapunov exponents, used to measure the exponential deviations from the given trajectory. On the other hand, the local stability of solutions of dynamical systems is much less understood. Even though the method of Lyapunov has been well studied, it remains interesting to extend the theory and methodology for stability analysis of dynamical systems from different points of view and to compare the results with those of the corresponding analysis of Lyapunov exponents. One of the approaches that offer important new insights for stability analysis of dynamical systems is what may be called the geometro-dynamical approach, which has been developed since Kosambi \cite{Kosambi1933}, Cartan \cite{Cartan1933} and Chern \cite{Chern1939}. The differential-geometric theory of Kosambi--Cartan--Chern (KCC) is established from the variational equations for the deviation of the whole trajectory to nearby ones, based on the existence of a one-to-one correspondence between a second-order dynamical system and the geodesic equations in an associated Finsler space (see \cite{BHS2012} for more details). The qualitative analysis of dynamical systems in practice usually involves heavy symbolic computation. In the last three decades, remarkable progress has been made on the research and development of symbolic and algebraic algorithms and software tools, bringing the classical qualitative theory of differential equations to computational approaches for symbolic analysis of the qualitative behaviors of diverse dynamical systems (see \cite{HLS97,WDMXBC2005,HTX15,dw91,EKW00,swek-mcs2009,vd09,adrt2016,acj17,HGR10,CCMYZ13,bhlr2015,bdetal:2020}, the survey article \cite{HNW2022}, and references therein).

In this paper, we study the Jacobi stability of dynamical systems of the form \eqref{eq1-0} and show how to effectively \textit{compute a partition of the parametric space of $\boldsymbol{\mu}$ such that inside every open cell of the partition, the system can have a prescribed number of Jacobi stable fixed points}. The main techniques we use are based on the KCC theory and advanced methods of polynomial and semi-algebraic system solving with exact symbolic computation. The idea of using the KCC theory to study the Jacobi stability of nonlinear dynamical systems is not new and a considerable amount of work has already been done on the subject (see \cite{PLA2003,SVS12005,SVS22005,BHS2012,HHLY2015,HPPS2016,GY2017} and references therein). However, the exploration of algorithmic and symbolic-computational approaches for systematical analysis of the Jacobi stability of dynamical systems is a new direction of research in which we are involved.

The rest of the paper is structured as follows. In Section \ref{sect2} some basic concepts and results of the KCC theory are reviewed. Section \ref{sect-main} presents the main results stated as Proposition \ref{kcc-prop-1} and Theorem \ref{semi-kcc}: the former leads to a decision procedure for the existence of Jacobi stable fixed points by quantifier elimination, and the latter allows one to derive necessary and sufficient conditions for systems of the form \eqref{eq1-0} to have given numbers of Jacobi stable fixed points. Two algorithmic schemes for Jacobi stability analysis are presented in Section \ref{sect4}. The effectiveness of our computational approach is demonstrated in Section \ref{sect5} with several examples including the famous Brusselator model, the Cdc2-Cyclin B/Wee1 system and the Lorenz--Stenflo system. The paper concludes in Section \ref{sect6} with some discussions on future research.

\section{KCC Theory and Jacobi Stability of Dynamical Systems}\label{sect2}

In this section, we briefly review the KCC theory and recall the basic concepts and notations that will be used later on. Some geometric objects and useful theoretical results are presented. For more details about the KCC theory, the reader may consult the expository article \cite{BHS2012} and other recent work in \cite{PLA2003,SVS12005,SVS22005,HPPS2016}.

Let $\mathcal{M}$ be a real, smooth $n$-dimensional manifold and let $\mathcal{TM}$ be its tangent bundle. On an open connected subset $\Omega$ of the Euclidean $(2n+1)$-dimensional space $\mathbb{R}^n\times\mathbb{R}^n\times\mathbb{R}$ we introduce a $(2n+1)$-dimensional coordinates system $(\boldsymbol{x},\boldsymbol{y},t)$, where $\boldsymbol{x}=(x_1,x_2,\ldots,x_n)$, $\boldsymbol{y}=(y_1,y_2,\ldots,y_n)$ and $t$ is the usual time coordinate. The coordinates $\boldsymbol{y}$ are defined as
\begin{equation}\label{kcc-1}
\begin{split}
\boldsymbol{y}=\left(\frac{d x_1}{d t},\frac{d x_2}{d t},\ldots,\frac{d x_n}{d t}\right).\nonumber
\end{split}
\end{equation}

A basic assumption in the KCC theory is that the time coordinate $t$ is an absolute invariant. Therefore the only admissible coordinate transformations are
\begin{equation}\label{kcc-2}
\begin{split}
\tilde{t}=t,\quad \tilde{x}_i=\tilde{x}_i(x_1,x_2,\ldots,x_n),\quad i\in\{1,2,\ldots,n\}.
\end{split}
\end{equation}
In many situations of scientific interest the equations of motion of a dynamical system can be derived from a Lagrangian function $L:\mathcal{TM}\rightarrow\mathbb{R}$ via the Euler--Lagrange equations:
\begin{equation}\label{kcc-3}
\begin{split}
\frac{d}{dt}\frac{\partial L}{\partial y_i}-\frac{\partial L}{\partial x_i}=F_i,\quad i=1,2,\ldots,n,
\end{split}
\end{equation}
where $F_i$, $i=1,2,\ldots,n$, is the external force. For a regular Lagrangian $L$, the Euler--Lagrange equations introduced in \eqref{kcc-3} are equivalent to a system of second-order ordinary differential equations
\begin{equation}\label{kcc-4}
\begin{split}
\frac{d^2x_i}{dt^2}+2G^i(\boldsymbol{x},\boldsymbol{y},t)=0,\quad i\in\{1,2,\ldots,n\},
\end{split}
\end{equation}
where each function $G^i(\boldsymbol{x},\boldsymbol{y},t)$ is $\mathcal{C}^{\infty}$ in a neighborhood of some initial conditions $(\boldsymbol{x}_0,\boldsymbol{y}_0,t_0)$ in $\Omega$.

The basic idea of the KCC theory is that if an arbitrary system of second-order differential equations of the form \eqref{kcc-4} is given, with no a priori Lagrangian function known, we can still study the behavior of its trajectories by using differential geometric methods. This study can be performed by using the close analogy with the trajectories of the standard Euler--Lagrange system.

For the non-singular coordinate transformations introduced through \eqref{kcc-2}, we define the KCC-covariant differential of a vector field $\xi_i(\boldsymbol{x})$ on the open subset $\Omega$ as
\begin{equation}\label{kcc-5}
\begin{split}
\frac{D\xi_i}{dt}=\frac{d\xi_i}{dt}+N_j^i\xi_j,\nonumber
\end{split}
\end{equation}
where $N_j^i=\partial G^i/\partial y_j$ are the coefficients of the nonlinear connection. For $\xi_i=y_i$, we obtain
\begin{equation}\label{kcc-6}
\begin{split}
\frac{Dy_i}{dt}=N_j^iy_j-2G^i=-\epsilon_i.\nonumber
\end{split}
\end{equation}
The contravariant vector field $\epsilon_i$ on $\Omega$ is called the first KCC invariant.

We vary now the trajectories $x_i(t)$ of the system \eqref{kcc-4} into nearby ones according to
\begin{equation}\label{kcc-7}
\begin{split}
\tilde{x}_i(t)=x_i(t)+\eta\xi_i(t),
\end{split}
\end{equation}
where $|\eta|$ is a small parameter, and $\xi_i(t)$ are the components of a contravariant vector field defined along the path $x_i(t)$. Substituting \eqref{kcc-7} into \eqref{kcc-4} and taking the limit $\eta\rightarrow0$ we obtain the deviation equations in the form
\begin{equation}\label{kcc-8}
\begin{split}
\frac{d^2\xi_i}{dt^2}+2N_j^i\frac{d\xi_j}{dt}+2\frac{\partial G^i}{\partial x_j}\xi_j=0.
\end{split}
\end{equation}
Equation \eqref{kcc-8} can be reformulated in the covariant form with the use of the KCC-covariant differential as
\begin{equation}\label{kcc-9}
\begin{split}
\frac{D^2\xi_i}{dt^2}=P_j^i\xi_j,
\end{split}
\end{equation}
where we have denoted
\begin{equation}\label{kcc-10}
\begin{split}
P_j^i=-2\frac{\partial G^i}{\partial x_j}-2G^{\ell}G^i_{j\ell}+y_{\ell}\frac{\partial N_j^i}{\partial x_{\ell}}+N_{\ell}^iN_j^{\ell}+\frac{\partial N_j^i}{\partial t},
\end{split}
\end{equation}
and we have introduced the Berwald connection $G^i_{j\ell}$, defined as
\begin{equation}\label{kcc-11}
\begin{split}
G^i_{j\ell}\equiv\frac{\partial N_j^i}{\partial y_{\ell}}.
\end{split}
\end{equation}
$P^i_j$ is called the second KCC-invariant or the deviation curvature tensor, while equation \eqref{kcc-9} is called the Jacobi equation. The third, fourth and fifth KCC-invariants are called respectively the torsion tensor, the Riemann--Christoffel curvature tensor, and the Douglas
tensor are defined as
\begin{equation}\label{kcc-12}
\begin{split}
P^i_{jk}=\frac{1}{3}\left(\frac{\partial P^i_j}{\partial y_k}-\frac{\partial P^i_k}{\partial y_j}\right),\quad P^i_{ik\ell}=\frac{\partial P^i_{jk}}{\partial y_{\ell}},\quad D^i_{jk\ell}=\frac{\partial G^i_{jk}}{\partial y_{\ell}}.\nonumber
\end{split}
\end{equation}
These tensors always exist in a Berwald space. In the KCC theory they describe the geometrical properties and interpretation of a system of second-order differential equations (see \cite{PLA2003,SVS12005,SVS22005} and references therein).

For many physical, chemical or biological applications one is interested in the behaviors of the trajectories of the dynamical system \eqref{kcc-4} in a vicinity of a point $x_i(t_0)$. For simplicity in the following we take $t_0=0$. We consider the trajectories $x_i=x_i(t)$ of \eqref{kcc-4} as curves in the Euclidean space $(\mathbb{R}^n,\langle\cdot,\cdot\rangle)$, where $\langle\cdot,\cdot\rangle$ represents the canonical inner product of $\mathbb{R}^n$, and suppose that the deviation vector $\xi$ satisfies $\xi(0)=O\in\mathbb{R}^n$, $\dot{\xi}(0)=W\neq O$; here, $O$ is the null vector.

For any two vectors $X,Y\in\mathbb{R}^n$ we define an adapted inner product $\langle\langle\cdot,\cdot\rangle\rangle$ to the deviation tensor $\xi$ by $\langle\langle X,Y\rangle\rangle:=1/\langle W,W\rangle\cdot\langle X,Y\rangle$. We also have $||W||^2:=\langle\langle W,W\rangle\rangle=1$. Thus, the focusing tendency of the trajectories around $t_0=0$ can be described as follows:
\begin{itemize}
  \item bunching together if $||\xi(t)||<t^2$, namely, if and only if the real part of the eigenvalues of $P_j^i(0)$ are strictly negative.
  \item disperses if $||\xi(t)||>t^2$, namely, if and only if the real part of the eigenvalues of $P_j^i(0)$ are strictly positive.
\end{itemize}

Based on the above considerations we introduce the rigorous definition
of the concept of Jacobi stability for a dynamical system as follows.
\begin{definition}\label{kcc-def}
If the system \eqref{kcc-4} of differential equations satisfies the initial conditions $||x_i(t_0)-\tilde{x}_i(t_0)||=0$, $||\dot{x}_i(t_0)-\dot{\tilde{x}}_i(t_0)||\neq0$, with respect to the norm $||\cdot||$ induced by a positive-definite inner product, then the trajectories of \eqref{kcc-4} are Jacobi stable if and only if the real
part of the eigenvalues of the deviation curvature tensor $P_j^i$ are strictly negative everywhere. Otherwise, the trajectories are Jacobi unstable.
\end{definition}

More detailed discussions of the KCC theory, including applications, can be found in \cite{PLA2003,BHS2012,HHLY2015,HPPS2016}.

\section{Theoretical Results on the Jacobi Stability}\label{sect-main}

In this section, we present our main results on the Jacobi stability of system \eqref{eq1-0}. The following lemma gives necessary and sufficient conditions for a fixed point of system \eqref{eq1-0} to be Jacobi stable. Its proof can be found in Appendix \ref{sec-A}.

\begin{lemma}\label{lem-main-1}
Let $\bar{\boldsymbol{x}}=(\bar{x}_1,\bar{x}_2,\ldots,\bar{x}_n)$ be an isolated fixed point of system \eqref{eq1-0}. Denoted by $J(\bar{\boldsymbol{x}},\boldsymbol{\mu})$ the Jacobian matrix of system \eqref{eq1-0} evaluated at the fixed point $\bar{\boldsymbol{x}}$. Then $\bar{\boldsymbol{x}}$ is Jacobi stable if and only if all the eigenvalues
of the matrix $J^2(\bar{\boldsymbol{x}},\boldsymbol{\mu})$ have negative real parts.
\end{lemma}

The following result is an immediate consequence of \cite[Theorem 2]{HPPS2016}, which shows that only \textit{even-dimensional} first-order dynamical systems can exhibit Jacobi stability.

\begin{theorem}\label{thm-kccmain-1}
If $n=2j+1$, $j\in\mathbb{N}$, then the fixed point $\bar{\boldsymbol{x}}$ of system \eqref{eq1-0} is always Jacobi unstable. If $n=2j$, then $\bar{\boldsymbol{x}}$ is Jacobi stable if and only if system \eqref{eq1-0} satisfies the following two conditions:
\begin{itemize}
  \item[(a)] all eigenvalues of the characteristic polynomial of Jacobian matrix $J(\bar{\boldsymbol{x}},\boldsymbol{\mu})$ are complex conjugate, i.e., $\lambda_j:=\alpha_j+i\beta_j$, $\bar{\lambda}_j:=\alpha_j-i\beta_j$, $j=1,\ldots,n/2$;
  \item[(b)] for any $j=1,\ldots,n/2$, we have $\alpha_j^2-\beta_j^2<0$.
\end{itemize}
\end{theorem}

\begin{proof}
The result follows directly from \cite[Theorem 2]{HPPS2016}. Here we only present the proof for the case $n=2j+1$ odd, more detailed arguments can be found in \cite[Section 4.2]{HPPS2016}. In fact, for \textit{odd-dimensional} first-order dynamical systems of the form \eqref{eq1-0}. The resulting characteristic polynomial of the Jacobian matrix $J(\bar{\boldsymbol{x}},\boldsymbol{\mu})$ has at least one real root. It follows from \cite[Lemma 2]{HPPS2016} that the matrix $J^2(\bar{\boldsymbol{x}},\boldsymbol{\mu})$ has at least one real positive root. That is, the fixed point $\bar{\boldsymbol{x}}$ is Jacobi unstable.
\end{proof}

\begin{remark}\label{rem-kcc-3-3}
In practice, the results established in Theorem \ref{thm-kccmain-1} are not suitable for application to real-world differential models. The main reason is that the root structure analysis of the characteristic polynomial of $J(\bar{\boldsymbol{x}},\boldsymbol{\mu})$ is very difficult, especially for parametric dynamical systems with higher dimensions. In this paper, we provide a systematic approach for analyzing the existence of Jacobi stability of dynamical systems by using the method of quantifier elimination automatically (see Proposition \ref{kcc-prop-1}). We also remark that for special systems of ODEs of the form \eqref{eq1-0}, one may reduce an \textit{odd-dimensional} first-order dynamical system to certain set of second-order differential equations; in such case an odd-dimensional first-order dynamical system could exhibit Jacobi stability without knowing the information on the matrix $J(\bar{\boldsymbol{x}},\boldsymbol{\mu})$, see \cite{HHLY2015,GY2016,MLTH2016,GY2017,HLL2019} and references therein.
\end{remark}

Our main question, relevant to the algorithmic aspects of Jacobi stability, then arises: \textit{under what conditions does a dynamical system of the form \eqref{eq1-0} have a prescribed number of Jacobi stable fixed points}? In the following we provide two algorithmic schemes to answer the question by means of symbolic and algebraic computation.

The first scheme addresses the existence of Jacobi stable fixed points, i.e., the problem of determining the conditions on the parameters $\boldsymbol{\mu}$ for a system of the form \eqref{eq1-0} to have at least one Jacobi stable fixed point. The next result is based on Theorem \ref{thm-kccmain-1}, which provides an algorithmic approach for testing Jacobi stability by quantifier elimination (QE). %The proof is detailed in Appendix \ref{sec-kccprop-B}.

\begin{proposition}\label{kcc-prop-1}
The problem of detecting the existence of Jacobi stable fixed points of a $2m$-dimensional dynamical system \eqref{eq1-0} can be formulated into the following QE problem
\begin{equation}\label{kcc-QEP}
\begin{split}
 \exists &b_1\cdots\exists b_m\exists c_1\cdots\exists c_m\exists \bar{\boldsymbol{x}}\\
&\,\Bigg[\left(\bigwedge_{i=1}^{2m}f_{i,1}(\bar{\boldsymbol{x}},\boldsymbol{\mu})=0\right) \,\,\bigwedge\left(\bigwedge_{i=1}^{2m}f_{i,2}(\bar{\boldsymbol{x}},\boldsymbol{\mu})\neq0\right)\\
&\,\,\bigwedge\left(\bigwedge_{j=1}^{m}a_{j}(\bar{\boldsymbol{x}},\boldsymbol{\mu})-a_j(b_1,\ldots,b_m;c_1,\ldots,c_m)=0\right)\\
&\,\,\bigwedge\left(\bigwedge_{j=1}^{m}2c_j-b_j^2>0\right)\Bigg],
\end{split}
\end{equation}
where $f_{i,1}(\bar{\boldsymbol{x}},\boldsymbol{\mu})$ and $f_{i,2}(\bar{\boldsymbol{x}},\boldsymbol{\mu})$ are respectively the numerator and denominator of the function $f_{i}(\bar{\boldsymbol{x}},\boldsymbol{\mu})$ in \eqref{eq1-0}.
\end{proposition}

\begin{proof}
As indicated by Theorem \ref{thm-kccmain-1}, we assume that $n=2m$ for some $m\in\mathbb{N}_{\ge 0}$. The characteristic polynomials $p(\lambda)$ of the matrix $J(\bar{\boldsymbol{x}},\boldsymbol{\mu})$ can be written as
\begin{equation}\label{kcc-p-eq1}
p(\lambda)=\lambda^{2m}+a_1(\bar{\boldsymbol{x}},\boldsymbol{\mu})\lambda^{2m-1}+\cdots+a_{2m}(\bar{\boldsymbol{x}},\boldsymbol{\mu}).
\end{equation}
Then all eigenvalues of $p$ are complex conjugates if and only if there exist $b_j$, $c_j$ for $j=1,\ldots,m$ such that
\begin{equation}\label{kcc-p-eq2}
p(\lambda)=\prod_{j=1}^m(\lambda^2+b_j\lambda+c_j),
\end{equation}
where
\begin{equation}\label{kcc-p-eq3}
b_j^2-4c_j<0.
\end{equation}
Taking the expansion of \eqref{kcc-p-eq2} and comparing its coefficients in terms of $\lambda$ with those in \eqref{kcc-p-eq1}, we may get a series of constraints on $b_j$, $c_j$ which can be simply written as \begin{equation}\label{kcc-p-eq4}
a_i(\bar{\boldsymbol{x}},\boldsymbol{\mu})=a_i(b_1,\ldots,b_m;c_1,\ldots,c_m),\quad i=1,\ldots,2m.\nonumber
\end{equation}

We further assume that $\alpha_j\pm\beta_j i$ are the conjugate roots of $\lambda^2+b_j\lambda+c_j=0$. More explicitly, we have
$\alpha_j=-b_j/2$ and $\beta_j=\sqrt{4c_j-b_j^2}/2$. The substitution of the two relations into $\alpha_j^2<\beta_j^2$ immediately yields
\begin{equation}\label{kcc-p-eq5}
2c_j-b_j^2>0.
\end{equation}
Obviously, the condition \eqref{kcc-p-eq3} is implied by the condition \eqref{kcc-p-eq5} and thus we can omit it. Collecting all the constraints on $b_j$'s and $c_j$'s and using Theorem \ref{thm-kccmain-1}, we formulate the desired QE problem \eqref{kcc-QEP}.
\end{proof}

\begin{remark}\label{kcc-rem-qe}
By Tarski theory, the QE problem of the first-order theory is solvable over the real field \cite{Tarski:51}. In other words, for any formula constructed from polynomial equations and inequalities over the  real field by logical connectives (i.e., $\wedge$ (and), $\vee$ (or), $\neg$ ) and quantifiers (i.e., $\forall$, $\exists$), one can find a quantifier-free equivalence to the given formula. The first practical algorithm for QE is based on Cylindrical Algebraic Decomposition (CAD), which is proposed by Collins in 1975 \cite{Collins:75b}. Following up his work, numerous improvements have been made since then (e.g., to cite a few \cite{Hong:90a,BPR98,McCallum:99,Brown_Gross:2006,Hong_Safey_El_Din:2012,Chen_Maza:2016}) and several software tools are implemented, such as \textsf{QEPCAD B} \cite{qepcadb} (an improvement of \textsf{QEPCAD } \cite{Collins_Hong:91}),
functions \textsf{Resolve} and \textsf{Reduce} in Mathematica (as described in \cite{AS:2010}), package REDLOG in \textsf{REDUCE} (implemented on the basis of the methods of CAD \cite{Collins:75b,Collins_Hong:91} and virtual term substitution \cite{VW:1994,VW:1997}), and the Maple package \textsf{RegularChains} (developed by Moreno Maza and his coworkers).
\end{remark}

Our second algorithmic scheme is devoted to providing necessary and sufficient conditions on $\boldsymbol{\mu}$ for systems of the form \eqref{eq1-0} to have given numbers of Jacobi stable fixed points. According to Lemma \ref{lem-main-1}, we need to determine whether all the eigenvalues of $J^2(\bar{\boldsymbol{x}},\boldsymbol{\mu})$ have negative real parts. This can be done by using the stability criterion of Routh--Hurwitz \cite{RMAM1982,PLMT1985}. Let
\begin{equation}\label{kcc-main-JJ}
\begin{split}
\bar{p}(\lambda)=\lambda^n+\bar{a}_1(\bar{\boldsymbol{x}},\boldsymbol{\mu})\lambda^{n-1}+\cdots+\bar{a}_n(\bar{\boldsymbol{x}},\boldsymbol{\mu})
\end{split}
\end{equation}
be the characteristic polynomial of the matrix $J^2(\bar{\boldsymbol{x}},\boldsymbol{\mu})$. The Routh--Hurwitz criterion reduces the problem of determining the negative signs of the real parts of the eigenvalues of $J^2$ to the problem of determining the signs of certain coefficients $a_i$ of $\bar{p}(\lambda)$ and the signs of certain determinants $\Delta_j$ of matrices with $a_i$ or 0 as entries.

The necessary and sufficient conditions for $\bar{p}(\lambda)$ to have all solutions such that ${\rm{Re}}(\lambda)<0$ can be written as
\begin{equation}\label{kcc-main-12}
\begin{split}
&\bar{a}_n>0,\quad \Delta_1=\bar{a}_1>0,\quad \Delta_2=\left|
 \begin{matrix}
    \bar{a}_1& \bar{a}_3 \\
   1 & \bar{a}_2 \\
  \end{matrix}
  \right|>0,\\
 &\Delta_3=\left|
 \begin{matrix}
    \bar{a}_1& \bar{a}_3 & \bar{a}_5\\
   1 & \bar{a}_2 & \bar{a}_4 \\
   0 & \bar{a}_1 & \bar{a}_3 \\
  \end{matrix}
  \right|>0,\ldots,\\
  &\Delta_j=\left|
 \begin{matrix}
    \bar{a}_1& \bar{a}_3 & \cdots & \cdots\\
   1 & \bar{a}_2 & \bar{a}_4 & \cdots\\
   0 & \bar{a}_1 & \bar{a}_3 & \cdots\\
   0 & 1 & \bar{a}_2 & \cdots\\
   \cdots & \cdots & \cdots & \cdots\\
   0 & 0 & \cdots & \bar{a}_j\\
  \end{matrix}
  \right|>0,\ldots,\\
\end{split}
\end{equation}
for all $j=1,2,\ldots,n$. Here $\Delta_1,\ldots,\Delta_n$ are known as the \textit{Hurwitz determinants} of $\bar{p}(\lambda)$.

The following result gives necessary and sufficient conditions for a given system of the form \eqref{eq1-0} to have a prescribed number of Jacobi stable fixed points.
\begin{theorem}\label{semi-kcc}
For a general $n$-dimensional dynamical system of the form \eqref{eq1-0}, the system has exactly $k$ Jacobi stable fixed points if and only if the following semi-algebraic system
\begin{equation}\label{kcc-main-13}
\left\{\begin{array}{l}
f_{1,1}(\bar{\boldsymbol{x}},\boldsymbol{\mu})=f_{2,1}(\bar{\boldsymbol{x}},\boldsymbol{\mu})
=\cdots=f_{n,1}(\bar{\boldsymbol{x}},\boldsymbol{\mu})=0,\\ f_{i,2}(\bar{\boldsymbol{x}},\boldsymbol{\mu})\neq0,\quad i=1,\ldots,n,\\
\bar{a}_n(\bar{\boldsymbol{x}},\boldsymbol{\mu})>0,\quad\Delta_j(\bar{\boldsymbol{x}},\boldsymbol{\mu})>0,\quad j=1,\ldots,n\\
\end{array}\right.
\end{equation}
has exactly $k$ distinct real solutions with respect to the variable $\bar{\boldsymbol{x}}$, where $\boldsymbol{\mu}=(\mu_1,\ldots,\mu_p)$ are parameters appearing in \eqref{eq1-0}.
\end{theorem}

\begin{proof}
The result follows directly from Lemma \ref{lem-main-1} and Routh--Hurwitz criterion.
\end{proof}

\begin{remark}\label{rem-semi-1}
In our approach, we do not explicitly compute the fixed points, nor the corresponding eigenvalues of $\bar{p}(\lambda)$. By Theorem \ref{semi-kcc}, the problem of Jacobi stability analysis is reduced to that of solving a semi-algebraic system (namely \eqref{kcc-main-13}). Theorem \ref{semi-kcc} provides a straightforward computational method to verify whether a given differential system has a prescribed number of Jacobi stable fixed points. Its main task is to find conditions on the parameters $\boldsymbol{\mu}$ for system \eqref{kcc-main-13} to have exactly $k$ distinct real solutions. In the next section, we will give a systematic approach for solving semi-algebraic systems and analyzing Jacobi stability of dynamical systems by using symbolic computation methods automatically.
\end{remark}

\section{Algorithmic Analysis of the Jacobi Stability}\label{sect4}
Based on the theoretical results established in Section \ref{sect-main}, we present our algorithmic tests {\bf Jacobi Test I} and {\bf Jacobi Test II}. Given an input dynamical system of the form \eqref{eq1-0}, these algorithmic tests verify its Jacobi stability by using QE formula \eqref{kcc-QEP} and solving semi-algebraic system \eqref{kcc-main-13}, respectively. In so doing, the first test is applicable only to verify the existence of Jacobi stable fixed points. The second test admits one to determine conditions on the parameters such that the given system has a prescribed number of Jacobi stable fixed points.

\subsection{Jacobi Test I}
The process of using the first algorithmic test for analyzing the existence of Jacobi stability of a given dynamical system can be divided into three parts.

{\bf Part 1:} Compute the Jacobian matrix $J(\bar{\boldsymbol{x}},\boldsymbol{\mu})$ of a given dynamical system of the form \eqref{eq1-0} and write the characteristic polynomial $p(\lambda)$ of $J$.

{\bf Part 2:} Formulate a QE problem using \eqref{kcc-QEP}.

{\bf Part 3:} Solve the problem by QE tools.

It is well known that real QE can be carried out algorithmically due to the ground-breaking work of Tarski \cite{Tarski:51}, Collins \cite{Collins:75b}, and others. See \cite{Hong:90a,BPR98,McCallum:99,Brown_Gross:2006,Hong_Safey_El_Din:2012,Chen_Maza:2016} for further information. In principle, the above problem in {\bf Part 3} can be carried out automatically by using QE algorithms. However, the symbolic computation involved is so huge that automatic synthesis is practically impossible with currently available QE software.

We remark that the problem of detecting the existence of Jabobi stability of a given dynamical system, i.e., the quantifier elimination problem (see \eqref{kcc-QEP}) may be formulated as the problem of determining the conditions on the parameters $\boldsymbol{\mu}$ for the following semi-algebraic system
\begin{equation}\label{kcc-main-test-1}
\left\{\begin{array}{l}
f_{1,1}(\bar{\boldsymbol{x}},\boldsymbol{\mu})=f_{2,1}(\bar{\boldsymbol{x}},\boldsymbol{\mu})
=\cdots=f_{2m,1}(\bar{\boldsymbol{x}},\boldsymbol{\mu})=0,\\ f_{i,2}(\bar{\boldsymbol{x}},\boldsymbol{\mu})\neq0,\quad i=1,\ldots,2m,\\
a_{j}(\bar{\boldsymbol{x}},\boldsymbol{\mu})-a_j(b_1,\ldots,b_m;c_1,\ldots,c_m)=0,\\
2c_j-b_j^2>0,\quad j=1,2,\ldots,m
\end{array}\right.
\end{equation}
to have at least one real solution. The general approach for solving such semi-algebraic systems is explained in detail in \cite{WDMXBC2005,niuwang08} and it is based on the methods of real solution classification \cite{LYBX05} and discriminant varieties \cite{DLFR}.

In view of the above remark, in the algorithmic analysis of Jacobi stability, our attention will be focused on the algorithmic steps for solving semi-algebraic systems (see details in Section \ref{sec-kcc-test2}).

\subsection{Jacobi Test II}\label{sec-kcc-test2}

Our purpose is to derive conditions on the parameters for systems of the form \eqref{eq1-0} to have given numbers of Jacobi stability fixed points. In what follows we present a general algorithmic approach for automatically analyzing Theorem \ref{semi-kcc} by using methods from symbolic computation. This approach is based on the one for studying Lyapunov stability of biological systems proposed by Wang, Xia and Niu \cite{WDMXBC2005,niuwang08}. The main steps of our computational approach are described as follows.

\smallskip

{\bf STEP 1}. Formulate the semi-algebraic system \eqref{kcc-main-13} from a dynamical system of the form \eqref{eq1-0}. Denote by $\mathcal{S}$ the semi-algebraic system for solving, $\Psi$ the set of inequalities of $\mathcal{S}$, $\mathcal{F}$ the set of polynomials in $\Psi$, and $\mathcal{P}$ the set of polynomials in the equations of $\mathcal{S}$.

\smallskip

{\bf STEP 2}. Triangularize the set $\mathcal{P}$ of
polynomials to obtain one or several (regular) triangular sets
$\mathcal{T}_{k}$ by using the method of triangular decomposition or
Gr\"obner bases.

\smallskip

{\bf STEP 3}. For each triangular set $\mathcal{T}_{k}$, use the
polynomial set $\mathcal{F}$ to compute an algebraic variety $V$ in
$\boldsymbol{\mu}$ by means of real solution classification (using, e.g.,
Yang--Xia's method \cite{YHX01,LYBX05} or Lazard--Rouillier's method
\cite{DLFR}), which decomposes the parameter space $\mathbb{R}^{p}$ into finitely many cells such that in each cell the number of real zeros of $\mathcal{T}_{k}$ and the signs of polynomials in $\mathcal{F}$ at these real zeros remain invariant. The algebraic variety is defined by polynomials in $\boldsymbol{\mu}$. Then take a rational sample point from each cell by using the method of PCAD or critical points \cite{msed2007}, and isolate the real zeros of $\mathcal{T}_{k}$ by rational intervals at this sample point. In this way, the number of real zeros of $\mathcal{T}_{k}$ and the signs of polynomials in $\mathcal{F}$ at these real zeros in each cell are determined.

\smallskip

{\bf STEP 4}. Determine the signs of (the factors of) the defining polynomials of $V$ at each sample point. Formulate the conditions on $\boldsymbol{\mu}$ according to the signs of these defining polynomials at the sample points in those cells in which the system $S$ has exactly the number of real solutions we want.

\smallskip

{\bf STEP 5}. Output the conditions on the parameter $\boldsymbol{\mu}$ such that the dynamical system has a prescribed number of Jacobi stable fixed points.

\begin{remark}\label{kcc-rem-step}
There are several software packages which can be used to realize the algorithmic steps in our approach. For example, the method of discriminant varieties of Lazard and Rouillier \cite{DLFR} (implemented as a Maple package {\sf DV} by Moroz and Rouillier), and the Maple package {\sf DISCOVERER} (see also recent improvements in the Maple package RegularChains[SemiAlgebraicSetTools]), developed by Xia, implements the methods of Yang and Xia \cite{LYBX05} for real solution classification. In Section \ref{sect5}, we will present several examples to demonstrate the applicability and the computational efficiency of our general algorithmic approach.
\end{remark}

\section{Examples and Experiments}\label{sect5}
In this section, we explain how to apply the algorithmic tests to study Jacobi stability of dynamical systems and illustrate some of the computational steps by a famous Brusselator model. In addition, using our computational approach, we also analyze Jacobi stability of the Cdc2-Cyclin B/Wee1 system and the Lorenz--Stenflo system. The experimental results show the applicability and efficiency of our algorithmic approach. All the experiments were made in Maple 17 on a Windows 10 laptop with 4 CPUs 2.9GHz and 8192M RAM.

\subsection{Illustrative Example: The Brusselator}\label{sect5.1}
The Brusselator is a simple biological system proposed by Prigogine and
Lefever in 1968 \cite{IPRL1968}. The chemical reactions are:
\begin{equation}\label{kcc-ex-1}
\begin{split}
A&\mathop{\longrightarrow}\limits^{k_1} X,\\
B+X&\mathop{\longrightarrow}\limits^{k_2} Y+D,\\
2X+Y&\mathop{\longrightarrow}\limits^{k_3} 3X,\\
X&\mathop{\longrightarrow}\limits^{k_4} E,
\end{split}
\end{equation}
where the $k_i$'s are rate constants, and the reactant concentrations of $A$ and $B$ are kept constant. Then, the law of mass action leads to the following dynamical system:
\begin{equation}\label{kcc-ex-2}
\begin{split}
\frac{dx}{dt}&=1-(b+1)x+ax^2y,\\
\frac{dy}{dt}&=bx-ax^2y,
\end{split}
\end{equation}
where $x$ and $y$ correspond to the concentrations of $X$ and $Y$, respectively, and $a$, and $b$ are positive constants.

Here we use algorithmic tests to study the Jacobi stability of the Brusselator model. The experimental results show as well the applicability and efficiency of our algorithmic approach. The fixed point of \eqref{kcc-ex-2} satisfies the following algebraic system:
\begin{equation}\label{kcc-ex-3}
\begin{split}
\{f_1&=1-(b+1)x+ax^2y=0,\\
 f_2&=bx-ax^2y=0,\,\, a>0,\,\,b>0\}.
\end{split}
\end{equation}
Simple computations show that the Jacobian evaluated at the fixed point $\bar{\boldsymbol{x}}=(x,y)$ is given by
\begin{equation}\label{kcc-ex-4}
\begin{split}
J(\bar{\boldsymbol{x}})=\left(
 \begin{matrix}
    2\,axy-b-1 & ax^2 \\
   -2\,axy+b & -ax^2 \\
  \end{matrix}
  \right),
\end{split}
\end{equation}
and hence the characteristic polynomials of $J(\bar{\boldsymbol{x}})$ and $J^2(\bar{\boldsymbol{x}})$ can be written as
\begin{equation}\label{kcc-ex-5}
\begin{split}
p(\lambda)&=\lambda^{2}+a_1(\bar{\boldsymbol{x}})\lambda+a_2(\bar{\boldsymbol{x}}),\\
\bar{p}(\lambda)&=\lambda^{2}+\bar{a}_1(\bar{\boldsymbol{x}})\lambda+\bar{a}_2(\bar{\boldsymbol{x}}),
\end{split}
\end{equation}
where
\begin{equation}\label{kcc-ex-5-5}
\begin{split}
a_1(\bar{\boldsymbol{x}})&=ax^{2}-2\,axy+b+1,\\
\bar{a}_1(\bar{\boldsymbol{x}})&=-a^{2}x^{4}+4a^{2}x^{3}y-4a^{2}x^{2}y^{2}-2abx^{2}\\
&\quad+4abxy+4axy-b^{2}-2b-1,\\
a_2(\bar{\boldsymbol{x}})&=ax^2,\quad\bar{a}_2(\bar{\boldsymbol{x}})=a^{2}x^{4}.\nonumber
\end{split}
\end{equation}

\noindent\underline{{\bf Jacobi Test I:}} Using Proposition \ref{kcc-prop-1}, we obtain the following QE problem
\begin{equation}\label{kcc-QEP-ex}
\begin{split}
\exists b_1&\exists c_1\exists \bar{\boldsymbol{x}}\\
&\Big[f_1=1-(b+1)x+ax^2y=0,\\
&\,\,f_2=bx-ax^2y=0,\\
&\,\,a_{1}(\bar{\boldsymbol{x}})-b_1=ax^{2}-2\,axy+b+1-b_1=0,\\
&\,\,a_{2}(\bar{\boldsymbol{x}})-c_1=ax^{2}-c_1=0,\\
&\,\,a>0,\,\,b>0,\,\,2c_1-b_1^2>0\Big].
\end{split}
\end{equation}
Using the Maple package RegularChains[SemiAlgebraicSetTools] to the above system, we obtain an equivalent quantifier-free formula $R_0=(a-b)^2-2\,b+1<0$. Thus we have the following result on the Jacobi stability of the Brusselator model.
\begin{theorem}\label{main-kcc-bm}
The Brusselator model \eqref{kcc-ex-2} has at least one Jacobi stable fixed point if $R_0=(a-b)^2-2\,b+1<0$.
\end{theorem}

\noindent\underline{{\bf Jacobi Test II:}} Using Theorem \ref{semi-kcc}, we obtain the following semi-algebraic system, which can be used for analyzing whether the Brusselator model \eqref{kcc-ex-2} has a prescribed number of Jacobi stable fixed points:
\begin{equation}\label{kcc-ex-6}
\mathcal{S}:\left\{\begin{array}{l}
f_1=1-(b+1)x+ax^2y=0,\\
f_2=bx-ax^2y=0,\\
\bar{a}_1(\bar{\boldsymbol{x}})>0,\quad
\bar{a}_2(\bar{\boldsymbol{x}})>0,\quad a>0,\quad b>0.
\end{array}\right.
\end{equation}
Solving the semi-algebraic system $\mathcal{S}$ by means of symbolic computation presented in Section \ref{sec-kcc-test2}, we see that system $\mathcal{S}$ has exactly one real solution if and only if $R_0
=(a-b)^2-2\,b+1<0$; system $\mathcal{S}$ has no real solution if and only if $R_0>0$.

From the above analysis, we have the following result on the Jacobi stability of the Brusselator model.
\begin{theorem}\label{main-kcc1}
The Brusselator model \eqref{kcc-ex-2} has exactly one or no Jacobi stable fixed point if and only if $R_0=(a-b)^2-2\,b+1<0$ or $R_0>0$.
\end{theorem}

Theorem \ref{main-kcc1} is consistent with the result in \cite[Section 3.3]{BHS2012}. In Figure \ref{fig:kcc-equilibria}, we provide partitions of the parameter set $\{(a,b)\,|\,a,b>0\}$ of \eqref{kcc-ex-2} for distinct numbers of Jacobi stable fixed points.

\begin{figure}[htbp]
    \centering
    \includegraphics[width=0.4\textwidth]{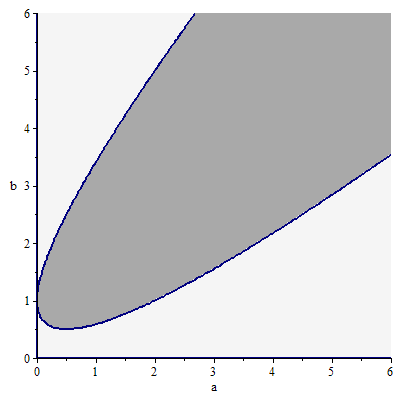}
    \caption{In the dark-gray and light-gray regions, system \eqref{kcc-ex-2} has one and zero Jacobi stable fixed point, respectively. The blue curve is defined by $R_0=0$.}\label{fig:kcc-equilibria}
\end{figure}

\subsection{Other Models and Remarks}\label{sect5.2}

\subsubsection{The Cdc2-Cyclin B/Wee1 System}\label{sect5.2-1}

The aim of this subsection is to study the Jabobi stability of the Cdc2-cyclin B/Wee1 system. We refer to \cite{AFS2004} for the setting details of this example. Under certain assumptions, the system of differential equations that model the two-component, mutually inhibitory feed-back loop is reduced to the following form
\begin{equation}\label{kcc-other-1}
\begin{split}
\frac{dx}{dt}&=\alpha_1(1-x)-\frac{\beta_1x(\upsilon y)^{\gamma_1}}{k_1+(\upsilon y)^{\gamma_1}},\\
\frac{dy}{dt}&=\alpha_2(1-y)-\frac{\beta_2yx^{\gamma_2}}
{k_2+x^{\gamma_2}},\\
\end{split}
\end{equation}
where $\alpha_1$, $\alpha_2$, $\beta_1$,$\beta_2$ are rate constants, $k_1$, $k_2$ are Michaelis constants, $\gamma_1$, $\gamma_2$ are Hill coefficients, and $\upsilon$ is a coefficient that reflects the strength of the influence of Wee1 on Cdc2-cyclin B. For easy reference and comparison, we take some of the parameter values for the biological constants suggested in \cite{AFS2004}:
\begin{equation}
\begin{split}
\alpha_1=\alpha_2=1,\,\, \beta_1=200,\,\, \beta_2=10,\,\,
k_1=30,\,\, \upsilon=24.\nonumber
\end{split}
\end{equation}
Then system \eqref{kcc-other-1} becomes
\begin{equation}\label{kcc-other-2}
\begin{split}
\frac{dx}{dt}&=1-x-200\,\frac{x\left(24\,y\right)^{\gamma_{{1}}}}{30+ \left(24\,y\right)^{\gamma_{{1}}}},\\
\frac{dy}{dt}&=1-y-10\,\frac{yx^{\gamma_{{2}}}}{k_{{2}}+x^{\gamma_{{2}}}}.
\end{split}
\end{equation}

Here we summarize our computational results for $\gamma_1,\gamma_2\in\{1,2\}$, report some of our experiments with comparisons for the proposed two algorithmic tests of Jacobi stability analysis, and provide timing statistics in table form to show the performance of the two schemes.

\begin{theorem}\label{main-3-kcc}
The following statements hold for the Cdc2-cyclin B/Wee1
system \eqref{kcc-other-2}.

{\bf Case 1:} $\gamma_1=\gamma_2=1$. System \eqref{kcc-other-2} has no Jacobi stable fixed point.

{\bf Case 2:} $\gamma_1=1$, $\gamma_2=2$. System \eqref{kcc-other-2} has at least one Jacobi stable fixed point if $R_1<0$ by {\bf Jacobi Test I}; and has exactly one Jacobi stable fixed point if and only if $R_1<0$ by {\bf Jacobi Test II}.

{\bf Case 3:} $\gamma_1=2$, $\gamma_2=1$. Application of {\bf Jacobi Test I} shows that system \eqref{kcc-other-2} has at least one Jacobi stable fixed point if one of the following two conditions
\begin{equation}\label{kcc-other-3}
\begin{split}
\mathcal{C}_1&=[R_2<0,\, R_3<0,\, R_4\leq0,\, R_5\leq0],\\
\mathcal{C}_2&=[0<R_2,\, R_3<0,\, R_4\leq0,\, R_5\leq0,\,\,R_6\leq0]\nonumber
\end{split}
\end{equation}
holds. Application of {\bf Jacobi Test II} shows that system \eqref{kcc-other-2} has exactly one or two Jacobi stable fixed points if and only if $R_7<0$ or the condition
\begin{equation}\label{kcc-other-4}
\begin{split}
\mathcal{C}_3=[R_8<0,\, 0<R_9,\, 0\leq R_{10}]\nonumber
\end{split}
\end{equation}
holds.

{\bf Case 4:} $\gamma_1=\gamma_2=2$. The {\bf Jacobi Test I} cannot detect the existence of Jacobi stable fixed point in a reasonably short time! Application of {\bf Jacobi Test II} shows that system \eqref{kcc-other-2} has exactly one or two Jacobi stable fixed points if and only if $R_{11}<0$ or the condition
\begin{equation}\label{kcc-other-5}
\begin{split}
\mathcal{C}_4=[0\leq R_{12},\, 0<R_{13},\, R_{14}\leq 0,\, R_{15}\leq 0,\,R_{16}\leq 0]\nonumber
\end{split}
\end{equation}
holds.
\end{theorem}
Theorem \ref{main-3-kcc} can be proved by using calculations and arguments similar to those for the proof of Theorem \ref{main-kcc1}. The detailed proof is omitted here. The explicit expressions of $R_i$ are put in Appendix \ref{sect-ex}. The results in Table \ref{Tab-A} show the difference of computational times using the two algorithmic tests. The times marked in blue and red denote the entire computations of {\bf Jacobi Test II} for system \eqref{kcc-other-2} to have exactly one and two Jacobi stable fixed points, respectively.

\begin{table}[h]
\caption{Computational times (in seconds) of the {\bf Jacobi Test I/II} for system \eqref{kcc-other-2}.}\label{Tab-A}
\begin{center}
\begin{tabular}{ccc}
\toprule

  Cases & {\bf Jacobi Test I} & {\bf Jacobi Test II}\\ \hline
   1 & 0.009 & 0.020\\ \hline
   2 & 0.329 & 0.141 \\ \hline
   3 & 37.203 & \textcolor{blue}{0.156}/\textcolor{red}{1.547} \\ \hline
   4 & -- & \textcolor{blue}{0.203}/\textcolor{red}{51.156}\\
   \bottomrule
\end{tabular}
\end{center}
\end{table}

\begin{remark}\label{rem-kcc-cdc2}
From Table \ref{Tab-A}, one can see that the {\bf Jacobi Test I} can be applied to stability analysis only for dynamical systems involving a few (say, less than 4) parameters and variables. The {\bf Jacobi Test II} is expected to have a better performance for systems of higher dimension.
\end{remark}

\subsubsection{The Lorenz--Stenflo System}\label{sect5.2-ls}

In order to save space, the details of our results on the Jacobi stability of the Lorenz--Stenflo system are placed in Appendix \ref{sect-kcc-B}.

\section{Conclusion and Future Work}\label{sect6}
In this paper and for the first time, the problem of the Jacobi stability analysis is tackled algorithmically for an important class of systems of ODEs of arbitrary dimension. We reduce the problems of detecting the existence and analyzing the numbers of Jacobi stable fixed points of dynamical systems to algebraic problems, and propose two algorithmic schemes to verify Jacobi stability of dynamical systems by using methods of symbolic computation. The results of experiments performed demonstrate the effectiveness of the proposed algorithmic tests.

Our work also indicates that algebraic methods for Jacobi stability analysis are feasible only for systems of small or moderate size. For large systems, one may try model reduction methods (see, e.g., \cite{BLLM2011}) to reduce them to small ones. It would be interesting to employ our algorithmic schemes for analyzing the Jacobi stability of dynamical systems from different domains of science and engineering. The involved symbolic computation may easily become too heavy and intractable as the number and degree of equations and the number of system parameters increase. How to simplify and optimize the steps of symbolic computation in current algorithmic analysis of the Jacobi stability is a nontrivial question that remains for further investigation.

As noted in Remark \ref{rem-kcc-3-3}, one may reduce a dynamical system of the form \eqref{eq1-0} to a set of second-order differential equations as \eqref{kcc-4} and study the Jacobi stability of the reduced system without information about the matrix $J(\bar{\boldsymbol{x}},\boldsymbol{\mu})$ (see \cite{HHLY2015,GY2016,MLTH2016,GY2017,HLL2019} and references therein). How to combine differential elimination theory and other techniques to automate the model reduction process for Jacobi stability analysis and how to develop practical software tools to detect the Jacobi stability of dynamical systems are some of the standing questions for our research.

\section*{Acknowledgments}
Huang's work is supported by the National Natural Science Foundation of China (NSFC 12101032 and NSFC 12131004), Yang's work is partially supported by NSFC 12326353.

\bibliographystyle{plain}
\bibliography{ref}

\appendix
\newpage

\section{Proof of Lemma \ref{lem-main-1}}\label{sec-A}

\begin{proof}
The proof can be done by using the ideas from \cite[Section 4]{HPPS2016}. It should be noted that our fixed point here has arbitrariness and is not limited to the origin.

By taking the derivative of system \eqref{eq1-0} with respect to the time parameter $t$, we obtain
\begin{equation}\label{kcc-main-1}
\begin{split}
\frac{d^2x_1}{dt^2}&=f_{11}(\boldsymbol{x},\boldsymbol{\mu})y_1+f_{12}(\boldsymbol{x},\boldsymbol{\mu})y_2+
\cdots+f_{1n}(\boldsymbol{x},\boldsymbol{\mu})y_n,\\
\frac{d^2x_2}{dt^2}&=f_{21}(\boldsymbol{x},\boldsymbol{\mu})y_1+f_{22}(\boldsymbol{x},\boldsymbol{\mu})y_2+
\cdots+f_{2n}(\boldsymbol{x},\boldsymbol{\mu})y_n,\\
&\vdots\\
\frac{d^2x_n}{dt^2}&=f_{n1}(\boldsymbol{x},\boldsymbol{\mu})y_1+f_{n2}(\boldsymbol{x},\boldsymbol{\mu})y_2+
\cdots+f_{nn}(\boldsymbol{x},\boldsymbol{\mu})y_n,\nonumber
\end{split}
\end{equation}
where
\begin{equation}\label{kcc-main-2}
\begin{split}
f_{ij}=\frac{\partial f_i(\boldsymbol{x},\boldsymbol{\mu})}{\partial x_j},\quad y_i=\frac{dx_i}{dt}, \quad i,j=1,\ldots,n.
\end{split}
\end{equation}
The system can be written as
\begin{equation}\label{kcc-main-3}
\begin{split}
\frac{d^2x_i}{dt^2}-\sum_{k=1}^nf_{ik}(\boldsymbol{x},\boldsymbol{\mu})y_k=0, \quad i=1,\ldots,n.
\end{split}
\end{equation}
By comparing equations \eqref{kcc-main-3} and \eqref{kcc-4} we find
\begin{equation}\label{kcc-main-4}
\begin{split}
G^i=-\frac{1}{2}\sum_{k=1}^nf_{ik}(\boldsymbol{x},\boldsymbol{\mu})y_k=-\frac{1}{2}J(\boldsymbol{x},\boldsymbol{\mu})\boldsymbol{y}^T, \quad i=1,\ldots,n,
\end{split}
\end{equation}
where $J(\boldsymbol{x},\boldsymbol{\mu})$ is the Jacobian matrix of system \eqref{eq1-0} evaluated at the point $\boldsymbol{x}=(x_1,x_2,\ldots,x_n)$.

A direct computation shows that the coefficients of the nonlinear connection $N_j^i$ in \eqref{kcc-5} are the following:
\begin{equation}\label{kcc-main-5}
\begin{split}
N_j^i=\frac{\partial G^i}{\partial y_j}=-\frac{1}{2}f_{ij}(\boldsymbol{x},\boldsymbol{\mu}), \quad i,j=1,\ldots,n,
\end{split}
\end{equation}
and hence
\begin{equation}\label{kcc-main-6}
\begin{split}
G^i_{j\ell}=\frac{\partial N_j^i}{\partial y_{\ell}}=0.\nonumber
\end{split}
\end{equation}
Then, for the components of the deviation curvature tensor $(P_j^i)$, given by \eqref{kcc-10}, we have
\begin{equation}\label{kcc-main-7}
\begin{split}
P_j^i=-2\frac{\partial G^i}{\partial x_j}+y_{\ell}\frac{\partial N_j^i}{\partial x_{\ell}}+N_{\ell}^iN_j^{\ell}.
\end{split}
\end{equation}
Then, with the use of \eqref{kcc-main-4} and \eqref{kcc-main-5} we obtain
\begin{equation}\label{kcc-main-8}
\begin{split}
P_j^i=\frac{1}{2}\sum_{\ell=1}^nf_{ij\ell}(\boldsymbol{x},\boldsymbol{\mu})y_{\ell}+
\frac{1}{4}\sum_{\ell=1}^nf_{i\ell}(\boldsymbol{x},\boldsymbol{\mu})f_{\ell j}(\boldsymbol{x},\boldsymbol{\mu}),
\end{split}
\end{equation}
where
\begin{equation}\label{kcc-main-9}
\begin{split}
f_{ij\ell}=\frac{\partial f_{ij}(\boldsymbol{x},\boldsymbol{\mu})}{\partial x_{\ell}}, \quad \ell=1,\ldots,n.\nonumber
\end{split}
\end{equation}

Evaluating $P_j^i$ at the fixed point $\bar{\boldsymbol{x}}=(\bar{x}_1,\bar{x}_2,\ldots,\bar{x}_n)$ we obtain
\begin{equation}\label{kcc-main-10}
\begin{split}
P=(P_j^i)\big|_{\boldsymbol{x}=\bar{\boldsymbol{x}}}=\frac{1}{4}\sum_{\ell=1}^nf_{i\ell}(\boldsymbol{x},\boldsymbol{\mu})f_{\ell j}(\boldsymbol{x},\boldsymbol{\mu})\big|_{\boldsymbol{x}=\bar{\boldsymbol{x}}}=\frac{1}{4}J^2(\bar{\boldsymbol{x}},\boldsymbol{\mu}).
\end{split}
\end{equation}
According to Definition \ref{kcc-def}, we complete the proof of Lemma \ref{lem-main-1}.
\end{proof}

\section{List of Expressions}\label{sect-ex}

The expressions of $R_i$ in Theorem \ref{main-3-kcc} are listed below:
\begin{equation}
\begin{split}
R_1&=70125304508199817427241\,k_{{2}}^{4}+2829746914856987864074\,k_{{2}}^{3}
+3361232488971186326\,k_{{2}}^{2}\\
&-32526719864490606\,k_{{2}}+98410095984901,\\
R_2&=34129852019334120275001\,k_{{2}}^{6}+18955466201762993305416\,k_{{2}}^{5}
+2079190417275568260115\,k_{{2}}^{4}\\&+11545127440608096320\,k_{{2}}^{3}
-3094243064633242585\,k_{{2}}^{2}\\&-15298201362610744\,k_{{2}}
+1392902183498461,\\
R_3&=2547433055974362750136440827901\,k_{{2}}^{10}+2549462537607706690564021423110\,k_{{2}}^{9}
\\&-11268940034759448427530866055\,k_{{2}}^{8}
+1771067395411006049507787720\,k_{{2}}^{7}\\&+28244030137459182635195610\,k_{{2}}^{6}
-460293255845386666548348\,k_{{2}}^{5}\\&+9100321029103822970810\,k_{{2}}^{4}
+297199449630260016520\,k_{{2}}^{3}+3645502337110575145\,k_{{2}}^{2}\\
&+33660749290927910\,k_{{2}}+169273934903501,\\
R_4&=530232127206073089345217425321878979119696099835597928551671071\,k_{{2}}^{22}\\
&+3097568644611971444018698881173599447265615683575058548497331842\,k_{{2}}^{21}\\
&+7631912584419224736955284085076914218830918666122263171845545206\,k_{{2}}^{20}\\
&+10243812154818022698633706855579005484350933865146851778358261390\,k_{{2}}^{19}\\
&+8031685712664902128506558699613174184490713998578705437363239640\,k_{{2}}^{18}\\
&+3608244453515524749277956778611111417229843502952863578657903764\,k_{{2}}^{17}\\
&+794764359535314512081976864785139140055413494553824640783948368\,k_{{2}}^{16}\\
&+20894102919957803979386118708497568327087562923158615870530244\,k_{{2}}^{15}\\
&-20077675696091265474887792731401874508822528862130335446054430\,k_{{2}}^{14}\\
&-2067712427622846475879664020713131922381294277091986666630480\,k_{{2}}^{13}\\
&-65182044845510285735335621012309574274177334661692302436484\,k_{{2}}^{12}\\
&+339635340361378887076606248053764520769182382620977531152\,k_{{2}}^{11}\\
&+61405280215488767442278117538451501738536091865851055396\,k_{{2}}^{10}\\
&+1006294203795016137037000243498516131331383390064323220\,k_{{2}}^{9}\\
&-5868935702170245109981453893271152154490719243964880\,k_{{2}}^{8}\\
&-249498164636329157855776482316565392957953809477476\,k_{{2}}^{7}\\
&-441175185347076290754389665091550246937940606057\,k_{{2}}^{6}\\
&+20236868784347726647226088401625126982040617134\,k_{{2}}^{5}\\
&+34188215606113191165438934804508226389357790\,k_{{2}}^{4}\\
&-827990633410769468615972272531952966688510\,k_{{2}}^{3}\\
&+765121082185867033046840058962291876076\,k_{{2}}^{2}
\\&+9150532230210533587885936821516051192\,k_{{2}}
\\&-19157158958607356689340105603878224,\\
R_5&=171111729271577451816916278294023637020390895326107388729132780846012918718939\,k_{{2}}^{27}\\
&+694058399451306261725309490847033737870051270843334752981725680436287342684503\,k_{{2}}^{26}\\
&+1178406123343676406775701066834141331856286513708627226487776933169942187601614\,k_{{2}}^{25}\\
&+1086184492602659587435877292825091909193535934090024046781151680622987684675950\,k_{{2}}^{24}\\
\nonumber
\end{split}
\end{equation}

\begin{equation}
\begin{split}
&+587600286465267683565422314814858949846965420943287191504340073072926730345075\,k_{{2}}^{23}\\
&+187696389921620346825367529671357600840225909847866759826485985846683732081095\,k_{{2}}^{22}\\
&+33823176214705151036374071791646083386183863063515016771850659515433885581640\,k_{{2}}^{21}\\
&+3168705954369816312044412433038541646244389993413871636586762301879685990920\,k_{{2}}^{20}\\
&+157719645153763614429585162767144417550947668046726238941183420945850858550\,k_{{2}}^{19}\\
&+3769263391289015682652935646237578780964143127713243504313067099864625550\,k_{{2}}^{18}\\
&+7781326543536335811696740024439763284918812524206846106330927985634740\,k_{{2}}^{17}\\
&-1598865380914933222287901070755314927309638217057705929494690264998220\,k_{{2}}^{16}\\
&-23611093059604333942073745159899456397702641168899437844647628990210\,k_{{2}}^{15}\\
&+288400340225816957967706406166422339125773748570149231817259494950\,k_{2}^{14}\\
&+8865162354955973036833080016831068975429805232497336344341656200\,k_{{2}}^{13}\\
&+8794551486034664260013873978797003612014468390045028510532040\,k_{{2}}^{12}\\
&-1185022490305343092089409689922446690952051454865289257511345\,k_{{2}}^{11}\\
&-5256101291995032544599277175058425930712185071813464946885\,k_{{2}}^{10}\\
&+70122573819275858453538672368351623032523271094492805550\,k_{{2}}^{9}\\
&+376338742984691754801480051134092494791792064666753550\,k_{{2}}^{8}\\
&-1316760022331698962634453099334594874797705155466705\,k_{{2}}^{7}\\
&-5505495750662436547862196117124969486088681961485\,k_{{2}}^{6}\\
&+12555177418738885902467941529691860825870238720\,k_{{2}}^{5}\\
&+176520205866143202283905174900805862536243200\,k_{{2}}^{4}\\
&-452038623179439358955980730065952248012800\,k_{{2}}^{3}\\
&-7697216805222090220411113153159403585536\,k_{{2}}^{2}\\
&+40208114346381897524553294486239920128\,k_{{2}}\\
&-55345609604109527952398613254209536,\\
R_6&=81642460060267214759804220641838072236440477070626195418905470941083810292\\
&968891926695412098418661614396677
\,k_{{2}}^{38}+8554635236181774812542851023404692010353\\
&26559787375208717798398612512610786885957329777463166218670134955276
\,k_{{2}}^{37}\\
&+3983641059892053868666048940241816372115146946526079749133366450143983\\
&685297060685693336628187138977843033906\,k_{{2}}^{36}+10800832180677194445434\\
&724051260565272028480117454678208524352048697677786324047502700362425\\
&134606637901537722\,k_{{2}}^{35}+1861184424818983364552609410428706085020953133599176487624\\
&1082756250300325354610515100540322258263356151997505
\,k_{{2}}^{34}\\
&+20433219326868882410780315108683964803514605147560796603099\\
&666960191437685382973750542489268548722637925121484
\,k_{{2}}^{33}\\
&+123524551562823807446935732021017843529952711743119791073497\\
&91543296837005788395904967344802498607798633008012
\,k_{{2}}^{32}\\
&-356101639163494904584678366530439864197827361569527075087357\\
&243781376483717571284745517388141968801812783988
\,k_{{2}}^{31}\\
 \nonumber
\end{split}
\end{equation}

\begin{equation}
\begin{split}
&-855543478581650557041675445854581230608181181114250639100152\\
&7367734013486422526195433784627330007312637521441
\,k_{{2}}^{30}\\
&-899222656124848590440998290217688465615415933561620143035182\\
&6589832463923763737244565570616466046466833094970
\,k_{{2}}^{29}\\
&-5362253387151727844716297094276345745694747840673658995214942\\
&893996067262880667596698680266828343090862736438
\,k_{{2}}^{28}\\
&-2084650284124464232269542587051075098782422361393046736570330\\
&946837060600316911112345284280063910005147443474
\,k_{{2}}^{27}\\
&-5331186684717390027527068690154122600161764949994032388574100\\
&35854388192435988534341897013194737657112671079
\,k_{{2}}^{26}\\
&-8453694258057299776118305307451377084831268341676793629271372\\
&2116878332568965543856244123775004196411212808
\,k_{{2}}^{25}\\
&-6832224634463696713909292711018776602397712379428983203235633\\
&498268994424148297257149070938223468138371800
\,k_{{2}}^{24}\\
&-1938698684014866374138286765755589817275141011508321565502230\\
&6295942470918376178704099463133702974373880
\,k_{{2}}^{23}\\
&+3886404714525680478436476571495300346854580041596559106229896\\
&1597227201888316054385408091911381863843985
\,k_{{2}}^{22}\\
&+2533430342443214931359132995183224473626578615631268777381320\\
&014363643942113962722417855508768218263010
\,k_{{2}}^{21}\\
&+2484480242851330232419560223316541759307083508048277092384835\\
&0211452828654708528020847014839342105970
\,k_{{2}}^{20}\\
&-3370307994157824138333596609165421189387369016903926676615333\\
&127742776429068021038429191536970118650
\,k_{{2}}^{19}\\
&-1210646592339666898297814541413731201419562890540753143608669\\
&28452459009344751647623530714108466405
\,k_{{2}}^{18}\\
&+6557212592303017019650067531434746769647057465109753392750866\\
&81671233076399250667462463825361260
\,k_{{2}}^{17}\\
&+1033195891630950864197923771406472172773858002506365960972264\\
&35022173723223610566637649700158860
\,k_{{2}}^{16}\\
&+1152618011724133911131881076353189139469520873389337311237604\\
&192111173479008723495933626296620
\,k_{{2}}^{15}\\
&-3296723261054974680220771950938953745369576124275012552208643\\
&8769811313475606078899676568675
\,k_{{2}}^{14}\\
&-7372432584851908452930560658218590879869355542842571780616650\\
&35527027726071680214565012358\,k_{{2}}^{13}+29427693514925808523437035512\\
&81367259799855177020979617678647821431517603375098235588246\,k_{{2}}^{12}\\
&+1652780329250025351870433200157445055533587654470086460801313\\
&55093078973901132224565826\,k_{{2}}^{11}+18320840656707481535820667046029\\
&1251465034441890044122232924617989283481554825270987\,k_{{2}}^{10}\\
 \nonumber
\end{split}
\end{equation}

\begin{equation}
\begin{split}
&-1892889107498379215759699674801737136844658974642339207583950\\
&3842752551576918808720\,k_{{2}}^{9}-2944866600091607446117968503796855809\\
&3063224064235059464395246348212593734580416\,k_{{2}}^{8}+1320191920665524\\
&107087294583146812989293475967326750339296236044198729718390912\,k_{{2}}^{7}\\
&-3068662117758780868064943909747958466085273140693944425723674\\
&92524770613138\,k_{{2}}^{6}-507716429982515482349455037867512887213746622\\
&11413983115410601506547565166\,k_{{2}}^{5}+140218391155856183217676031924\\
&298597312950491398447790358188744468264380\,k_{{2}}^{4}+56431135009290372\\
&6722954096659542967326729422312749562519733968190672\,k_{{2}}^{3}-32453573\\
&45340865168544329892341383787825462690287868803341183598144\,
k_{{2}}^{2}\\
&+4672541551069989116041573841580968722813107076376215392621660576\,k_{{
2}}\\
&-1145363252306785173851064975145467934264549208846368468611648,\\
R_7&=34129852019334120275001\,k_{{2}}^{6}+18955466201762993305416\,k_{{2}}^{5}\\
&+2079190417275568260115\,k_{{2}}^{4}+11545127440608096320\,k_{{2}}^{3}\\
&-3094243064633242585\,k_{{2}}^{2}-15298201362610744\,k_{{
2}}+1392902183498461,\\
R_8&=74878248801\,k_{{2}}^{4}+75073646804\,k_{{2}}^{3}-325811394\,k_{{2}}^{2}
-35967996\,k_{{2}}+491401,\\
R_9&=34129852019334120275001\,k_{{2}}^{6}+18955466201762993305416\,k_{{2}}^{5}\\
&+2079190417275568260115\,k_{{2}}^{4}+11545127440608096320\,k_{{2}}^{3}\\
&-3094243064633242585\,k_{{2}}^{2}-15298201362610744\,k_{{
2}}+1392902183498461,\\
R_{10}&=514959115623141851336003033444651672162268192381\,k_{{2}}^{13}\\
&+117109671411481100316406909482175799959340582595\,k_{{2}}^{12}\\
&-18340125718418870324370612366298647132416682954\,k_{{2}}^{11}\\
&+366527529613459370771331193768602701751101450\,k_{{2}}^{10}\\
&-379955788713767745549920982715959250364501965\,k_{{2}}^{9}\\
&+21889673872367143957751295631807636649008877\,k_{{2}}^{8}\\
&+341871929691167075836381668283836930880020\,k_{{2}}^{7}\\
&-51724631573689868994658484974779627274708\,k_{{2}}^{6}\\
&+1173833924252340563250903096402567527075\,k_{{2}}^{5}\\
&+5250100448198199403687806473352176285\,k_{{2}}^{4}\\
&-648976539257362445568036079320186954\,k_{{2}}^{3}\\
&+11661228914135151280020506988360650\,k_{{2}}^{2}\\
&-92712858325208352861180603850579\,k_{{2}}\\
&+289072187996578736260607421875.\\
R_{11}&=1078794467013971763259359619128104601\,k_{{2}}^{6}+
4190789390543288\\
&52718339349023527216\,k_{{2}}^{5}+33396015289931023560113178292575115\,k_{{2}}^{4}\\
&+51109179326295655736130695620320\,k_{{2}}^{3}+2612847963967947253253\\
&6013415\,k_{{2}}^{2}+5418875034949417696457856
\,k_{{2}}+533709272401908912261,\\
R_{12}&=47192996838389807416746580659\,k_{{2}}^{4}-
20966587660735653205580364\,k_{{2}}^{3}\\
&-658192566458176776046\,k_{{
2}}^{2}+563983250276411636\,k_{{2}}-400954223341,\\
\nonumber
\end{split}
\end{equation}

\begin{equation}
\begin{split}
R_{13}&=1078794467013971763259359619128104601\,k_{{2}}^{6}+41907893905432\\
&8852718339349023527216\,k_{{2}}^{5}+33396015289931023560113178292575115\,k_{{2}}^{4}\\
&+51109179326295655736130695620320\,k_{{2}}^{3}+261284796396794725325360\\
&13415\,k_{{2}}^{2}+5418875034949417696457856\,k_{{2}}+533709272401908912261,\\
R_{14}&=648918481217371346882003704556042428115795163\,k_{{2}}^{7}-120695\\
&20134716704004714952557737875934559559\,k_{{2}}^{6}+4418216183863088669\\
&646708188024293824411\,k_{{2}}^{5}-225642379665655687012727484781439055\,k_{{2}}^{4}\\
&-94217435382044199079650160444215\,k_{{2}}^{3}+
123328823049113171017254\\
&37603\,k_{{2}}^{2}-202932195376509203593919\,k_{{2}}+2141281744336142451,\\
R_{15}&=14056083607991138241973991322676251462172878957848237991\\
&3553898385129889746029544845326003185723\,k_{{2}}^{16}+2487625815782812\\
&35354806794696673165631650960435230785210418714255826514104406\\
&22889919214839006\,k_{{2}}^{15}-155897796352570630306233247711593095023\\
&6049804018682712620839206639305264934164189345916714152\,k_{{2}}^{14}+\\
&33048677661362277534579674865020142536299750624126833610998614\\
&687308501888978076461728818466\,k_{{2}}^{13}-3048362047040274113781939\\
&99924518110797455949109123936691399056447996404599640181367954\\
&976\,k_{{2}}^{12}+1066718616203081609059336541064799313529528292309053\\
&117793766851288330417409309962131390\,k_{{2}}^{11}-32299540207717384207897\\
&6177113703107750469071212584166178363676248683954792172388064\,k_{{2}}^{10}\\
&-4571663446818831633613643989738970986766239848463453468645\\
&90653471179878132193934\,k_{{2}}^{9}+927944105023302105425959038081905\\
&53340860402362044666722444639156785683061142\,k_{{2}}^{8}+301760908371\\
&49661510853942896789833165111854613937461288783334854353093578\,k_{{2}}^{7}\\
&-1211204553937845261047331963454850242156244631992369887070\\
&486269112584\,k_{{2}}^{6}-157293712612208298921791399071661576835299\\
&6009193827408523506011306\,k_{{2}}^{5}-10549702522586887516038904142\\
&925745391047495992912539169946232\,k_{{2}}^{4}+274332540684284987959\\
&35202777461547303123585416897513398426\,k_{{2}}^{3}+1460770671064342\\
&8657527700471030908560741272471074704\,k_{{2}}^{2}-1510216620258154\\
&60184527833765062030321382907626\,k_{{2}}+68022059837701619029533623932680489962439,\\
R_{16}&=46443086472084477153537655823845850437972818518257003503456\\
&394063472690608258502194893271995676780497161204331247203\,k_{{2}}^{19}\\
&-52158170657966182914162844526727943938490389931550890626039\\
&6478645497829435537399388053538565860982068438876535183\,k_{{2}}^{18}\\
&+3201257950756916754465804315535444189774519741953603309022\\
&201006215501288979507505209223781587857597403874816759\,k_{{2}}^{17}\\
&-1927995411081746058973956781918474223027984128328479003815\\
&127959362776009883822706578699596916599521844783923\,k_{{2}}^{16}\\
&-2268356906763940471268547817632435582187593709272285523121\\
\nonumber
\end{split}
\end{equation}

\begin{equation}
\begin{split}
&16167860486647589829033930636118167091647211276\,k_{{2}}^{15}\\
&+53584393800503414261090476532734443550844345600599565076920\\
&6857989232214394382033462855395684493550124\,k_{{2}}^{14}\\
&-19443976348590594703350908603230305113809875060581048140581\\
&6912862013245444213168220275974741047724\,k_{{2}}^{13}\\
&+28992180967930349314441595446461339819791158773334326427838\\
&411845390402370921204656323963954252\,k_{{2}}^{12}-28394928847\\
&94291756669617736291488403000782829994786262524426867742351\\
&17861143059182535614\,k_{{2}}^{11}-5773592163286104238291604956032\\
&91354259491139120371012380026364185450663400998521016538\,k_{{2}}^{10}\\
&+9947593079915175208553037303126970657054478671877805850024376\\
&7986765129498109929514\,k_{{2}}^{9}-8535714135086335438265639803987\\
&202183206550314600701159128643417292433587490274\,k_{{2}}^{8}+418329857\\
&38467958747560692134980030626881692592356322830329320\\
&5262311495252\,k_{{2}}^{7}-1261829187607242551912965107896606397635\\
&3591963391247363619554648675284\,k_{{2}}^{6}+1908800964866068184886\\
&43179802191139907804565129923180626416053652\,k_{{2}}^{5}-176042942\\
&0526840061516849181898454180692596708492743946208532\,k_{{2}}^{4}\\
&+3261775682688532894456623956405746916790234221598965267\,k_{{2}}^{3}\\
&-4220059061286457391465143087604107072714470954991\,k_{{2}}^{2}\\
&-82934020944978298855798955351132208803593\,k_{{2}}\\
&+301375413724781100070065356301069789.\nonumber
\end{split}
\end{equation}

\section{The Lorenz--Stenflo System}\label{sect-kcc-B}

In 1996, Stenflo \cite{LS1996} derived a system to describe the evolution of finite amplitude acoustic gravity waves in a rotating atmosphere. This system is rather simple and reduces to the well-known Lorenz system \cite{LEN1963} when the parameter associated with the flow rotation is set to zero. The Lorenz--Stenflo (L--S) system is described by
\begin{equation}\label{kcc-LS-1}
\begin{split}
\dot{x}&=a(y-x)+dw,\\
\dot{y}&=cx-y-xz,\\
\dot{z}&=-bz+xy,\\
\dot{w}&=-x-aw,
\end{split}
\end{equation}
where $a$, $b$, $c$, $d$ are real parameters; $a$, $c$, $d$ are the Prandtl, the Rayleigh, and the rotation numbers respectively, and $b$ is a geometric parameter.

Here we provide some explicit conditions for system \eqref{kcc-LS-1} to have Jacobi stable fixed points by using the {\bf Jacobi Test II}. The main result is stated as follows.

\begin{theorem}\label{main-4-kcc}
The L--S system \eqref{kcc-LS-1} has two Jacobi stable fixed points if one of the following three conditions
\begin{equation}\label{kccls-1}
\begin{split}
\mathcal{C}_1&=[c=26,d=3/2]\,\wedge\,[b<0,\, a<0,\, 0<T_1,\, 0<T_2],\\
\mathcal{C}_2&=[c=26,d=3/2]\,\wedge\,[0<b,\, a<0,\, 0<T_1,\, T_2<0],\\
\mathcal{C}_3&=[b=7/10,d=3/2]\,\wedge\,[a<0,\, 0<T_3,\, 0<T_4]\nonumber
\end{split}
\end{equation}
holds.
\end{theorem}

Theorem \ref{main-4-kcc} can be proved by using similar calculations and arguments to the proof of Theorem \ref{main-kcc1}. The detailed proof is omitted here. The explicit expressions of $T_i$ are the following:
\begin{equation}
\begin{split}
T_1&=64\,a^{15}+128\,a^{13}b^{2}+64\,a^{11}b^{4}+192\,a^{14}-8000\,a^{13}b-1344\,a^{12}b^{2}-8000\,a^{11}b^{3}\\
&\quad-1536\,a^{10}b^{4}+608\,a^{13}-16000\,a^{12}b+178848\,a^{11}b^{2}+
73600\,a^{10}b^{3}+21824\,a^{9}b^{4}\\
&\quad+1088\,a^{12}-12320\,a
^{11}b+198176\,a^{10}b^{2}-1097920\,a^{9}b^{3}-6720\,a^{8}b
^{4}+1760\,a^{11}\\
&\quad-59840\,a^{10}b+21296\,a^{9}b^{2}+201984\,a
^{8}b^{3}+32448\,a^{7}b^{4}+2352\,a^{10}+768\,a^{9}b\\
&\quad+585312\,a^{8}b^{2}+69504\,a^{7}b^{3}-9504\,a^{6}b^{4}+1776\,a^{9}
-53952\,a^{8}b-394512\,a^{7}b^{2}\\
&\quad+110016\,a^{6}b^{3}-432
\,a^{5}b^{4}+2736\,a^{8}-39600\,a^{7}b+399888\,a^{6}b^{2}+
15984\,a^{5}b^{3}\\
&\quad-3888\,a^{4}b^{4}-684\,a^{7}+25056\,a^{6}
b-34992\,a^{5}b^{2}-39744\,a^{4}b^{3}-324\,a^{3}b^{4}+2916
\,a^{6}\\
&\quad-86724\,a^{5}b-41796\,a^{4}b^{2}-972\,a^{3}b^{3}+
648\,a^{2}b^{4}-3618\,a^{5}+36504\,a^{4}b+3402\,a^{3}b^{2}\\
&\quad+1944\,a^{2}b^{3}+3564\,a^{4}-35154\,a^{3}b-486\,a^{2}b^{2
}-4212\,a^{3}-2268\,a^{2}b\\
&\quad+1215\,ab^{2}+2673\,a^{2}+2430\,ab-
1944\,a+729,\\
T_2&=1024\,a^{25}b+2048\,a^{23}b^{3}+1024\,a^{21}b^{5}-102400\,a^{25}-46080\,a^{24}b-332800\,a^{23}b^{2}\\
&-119808\,a^{22}b^{3}-230400\,a^{21}b^{4}-73728\,a^{20}b^{5}-409600\,a^{24}+
10800640\,a^{23}b+7833600\,a^{22}b^{2}\\
&+18090496\,a^{21}b^{3}
+9676800\,a^{20}b^{4}+3507200\,a^{19}b^{5}-1734656\,a^{23}+
37417984\,a^{22}b\\
&-424120832\,a^{21}b^{2}-312825344\,a^{20}b^
{3}-386302976\,a^{19}b^{4}-95136768\,a^{18}b^{5}\\
&-4225024\,a^
{22}+100739072\,a^{21}b-1135625216\,a^{20}b^{2}+7420781312\,a^
{19}b^{3}\\
&+4588603392\,a^{18}b^{4}+1110275584\,a^{17}b^{5}-9981952\,a^{21}+224365824\,a^{20}b\\
&-1114916352\,a^{19}b^{2}+10749888512\,a^{18}b^{3}-55741947392\,a^{17}b^{4}-421719552\,
a^{16}b^{5}\\
&-17460224\,a^{20}+308750848\,a^{19}b-4281228288\,a
^{18}b^{2}-2329730048\,a^{17}b^{3}\\
&+12866933760\,a^{16}b^{4}+
2085397248\,a^{15}b^{5}-27503616\,a^{19}+573111552\,a^{18}b\\
&-1304512512\,a^{17}b^{2}+34330615296\,a^{16}b^{3}-18127017216\,
a^{15}b^{4}-631109376\,a^{14}b^{5}\\
&-37567488\,a^{18}+
307343232\,a^{17}b-4333292544\,a^{16}b^{2}-26064317184\,a^{15}b^{3}\\
&+8521984512\,a^{14}b^{4}+829949184\,a^{13}b^{5}-37366272\,a^{17}+899697024\,a^{16}b\\
&-2002513536\,a^{15}b^{2}+27555901056\,a^{14}b^{3}-7730560512\,a^{13}b^{4}-343115136\,a
^{12}b^{5}\\
&-46872576\,a^{16}-213326784\,a^{15}b-489215232\,a^{
14}b^{2}-9876429504\,a^{13}b^{3}\\
&-95095296\,a^{12}b^{4}+154484928\,a^{11}b^{5}-12517632\,a^{15}+963038592\,a^{14}b-
944495424\,a^{13}b^{2}\\
&+2960556480\,a^{12}b^{3}+6986295360\,a
^{11}b^{4}-45748800\,a^{10}b^{5}-42667776\,a^{14}\\
&-531171648\,a^{13}b-793224576\,a^{12}b^{2}+314807904\,a^{11}b^{3}+301242240\,a^{10}b^{4}\\
&-190022112\,a^{9}b^{5}+41326848\,a^{13}+426272544\,a^{12}b+1781828928\,a^{11}b^{2}\\
&+137505600\,a^{10}b^{3}-559563552\,a^{9}b^{4}+2807136\,a^{8}b^{5}-50917248\,a^{12}\\
&+21117024\,a^{11}b-1688884992\,a^{10}b^{2}+2420528832\,a
^{9}b^{3}-815033664\,a^{8}b^{4}\\&+5252688\,a^{7}b^{5}+86220288\,a^{11}
-414950688\,a^{10}b+760757184\,a^{9}b^{2}\\
&-1845485856\,
a^{8}b^{3}-195352560\,a^{7}b^{4}+21216816\,a^{6}b^{5}-74929536\,a^{10}\\
&+479471076\,a^{9}b-53887680\,a^{8}b^{2}-353537784\,a^{7}b^{3}+228217824\,a^{6}b^{4}\\
&+4137804\,a^{5}b^{5}+94700016\,a^{9}-537666660\,a^{8}b-252085284\,a^{7}b^{2}\\
&+315240012\,a^{6}b^{3}+31589028\,a^{5}b^{4}-4181544\,a^{4}b^{5}-78708672\,a^{8}\\
&+193912542\,a^{7}b-195062904\,a^{6}b^{2}+20483442\,a^{5}b^{3}-21327624\,a^{4}b^{4}\\
&-314928\,a^{3}b^{5}+64857672\,a^{7}-120451212\,a^{6}b+150889878\,a^{5}b^{2}\\
\nonumber
\end{split}
\end{equation}

\begin{equation}
\begin{split}
&-8752374\,a^{4}b^{3}-1154736\,a^{3}b^{4}+209952\,a^{2}b^{5}
-47046744\,a^{6}-32866236\,a^{5}b\\
&+44203644\,a^{4}b^{2}-5045409
\,a^{3}b^{3}+629856\,a^{2}b^{4}+23304672\,a^{5}+16697745\,a
^{4}b\\
&-21454470\,a^{3}b^{2}-551124\,a^{2}b^{3}-11311164\,a^{4}+9789012\,a^{3}b-1915812\,a^{2}b^{2}\\
&+393660\,ab^{3}+2283228
\,a^{3}-2696571\,a^{2}b+787320\,ab^{2}+866052\,a^{2}\\
&-866052\,ab-236196\,a+236196\,b.\\
T_3&=160000\,a^{15}+480000\,a^{14}-560000\,a^{13}c+2236800\,a^{13}-1198400\,a^{12}c\\
&+352800\,a^{11}c^{2}+4232000\,a^{12}-1954400\,
a^{11}c+501760\,a^{10}c^{2}-54880\,a^{9}c^{3}\\
&+7428816\,a^{11}-5674256\,a^{10}c+172088\,a^{9}c^{2}+14456512\,a^{10}-
2222640\,a^{9}c\\
&+1505280\,a^{8}c^{2}+9533088\,a^{9}-9355752\,a
^{8}c-623868\,a^{7}c^{2}+24279072\,a^{8}\\
&-2295720\,a^{7}c+
1128960\,a^{6}c^{2}+6212880\,a^{7}-5944428\,a^{6}c-132300\,a^{5}c^{2}\\
&+21012912\,a^{6}-4093740\,a^{5}c+5641812\,a^{5}-949914\,a^{4}c+9875412\,a^{4}\\
&-2835000\,a^{3}c+4799979\,a^{3}+4174092\,a^{2}+880875\,a+1822500,\\
T_4&=102400000\,a^{25}c-120320000\,a^{25}+445440000\,a^{24}c-304640000\,a^{23}c^{2}\\
&-535040000\,a^{24}+2584832000\,a^{23}c-1232896000\,a^{22}c^{2}+338688000\,a^{21}c^{3}\\
&-2854553600\,a^{23}+7597952000\,a^{22}c-4368268800\,a^{21}c^{2}+1096345600\,a^{20}c^{3}\\
&-162444800\,a^{19}c^{4}-8232422400\,a^{22}+
21222730240\,a^{21}c-10897295360\,a^{20}c^{2}\\
&+1866045440\,a^{19}c^{3}-328401920\,a^{18}c^{4}+30732800\,a^{17}c^{5}-23273663232\,a^{21}\\
&+45899996416\,a^{20}c-19264251776\,a^{19}c^
{2}+5598612992\,a^{18}c^{3}-106028160\,a^{17}c^{4}\\
&-49983813632
\,a^{20}+83374773504\,a^{19}c-41432491520\,a^{18}c^{2}+3739328768\,a^{17}c^{3}\\
&-1082672640\,a^{16}c^{4}+18439680\,a^
{15}c^{5}-96326281728\,a^{19}+149143164544\,a^{18}c\\
&-37381367744
\,a^{17}c^{2}+12369205632\,a^{16}c^{3}+255323712\,a^{15}c^
{4}-167647532032\,a^{18}\\
&+173787017856\,a^{17}c-91354190592\,a^{
16}c^{2}+1789062912\,a^{15}c^{3}-1157748480\,a^{14}c^{4}\\
&+13829760\,a^{13}c^{5}-227425725312\,a^{17}+312896469312\,a^{16
}c-28797599040\,a^{15}c^{2}\\
&+17772822144\,a^{14}c^{3}+137606112
\,a^{13}c^{4}-362506999296\,a^{16}+180192326976\,a^{15}c\\
&-132395042304\,a^{14}c^{2}-4868809344\,a^{13}c^{3}-598631040\,
a^{12}c^{4}-304400165184\,a^{15}\\
&+467753832864\,a^{14}c+
7232649984\,a^{13}c^{2}+18692868096\,a^{12}c^{3}+118688976\,a
^{11}c^{4}\\
&-561297799296\,a^{14}+42548531040\,a^{13}c-
129109365504\,a^{12}c^{2}-7775853120\,a^{11}c^{3}\\
&-284497920\,
a^{10}c^{4}-182374282944\,a^{13}+504294917040\,a^{12}c+21884278248\,a^{11}c^{2}\\
&+11333685888\,a^{10}c^{3}+33339600\,a
^{9}c^{4}-663683825664\,a^{12}-79178271120\,a^{11}c\\
&-75280562784\,a^{10}c^{2}
-2945887056\,a^{9}c^{3}+58961684592\,a^{11}+
358509021624\,a^{10}c\\
&+1343542788\,a^{9}c^{2}+2332152648\,a^{8}
c^{3}-595091251008\,a^{10}-46680515784\,a^{9}c\\
&-16888045104\,a^
{8}c^{2}+514382400\,a^{7}c^{3}+170623493928\,a^{9}+
136412319996\,a^{8}c\\
&-11093173560\,a^{7}c^{2}-377823441168\,a^{
8}+23152074804\,a^{7}c+1877169168\,a^{6}c^{2}\\
&+111499520652\,a^{7}+15649706898\,a^{6}c-4208444100\,a^{5}c^{2}-164481701004\,a^{6}\\
&+21640434984\,a^{5}c-459270000\,a^{4}c^{2}+43753552839\,a^{5}+1566740556\,a^{4}c\\
&-60554338344\,a^{4}+2447253000\,a^{3}c+22319203239\,a^{3}+2755620000\,a^{2}c\\
&-22575718656\,a^{2}+8040505500\,a-4133430000.\nonumber
\end{split}
\end{equation}

\end{document}